\documentclass[12 pt]{article}
\usepackage{cite,amsmath,amsfonts,amsthm,fullpage}
\usepackage{youngtab}
\usepackage{cite,amsmath,amsfonts,amsthm,fullpage, amssymb, mathtools, url}
\usepackage[usenames,svgnames]{xcolor}
\usepackage{youngtab, enumitem, comment, mathdots, graphicx, float}
\usepackage[symbol]{footmisc}
\usepackage[pagebackref=false]{hyperref} 
\usepackage{tcolorbox}

\theoremstyle{plain}
\theoremstyle{definition}
\newtheorem{theorem}{Theorem}[section]
\newtheorem{lemma}[theorem]{Lemma}
\newtheorem{definition-theorem}[theorem]{Definition-Theorem}
\newtheorem{definition-proposition}[theorem]{Definition-Proposition}

\newtheorem{proposition}[theorem]{Proposition}
\newtheorem{corollary}[theorem]{Corollary}
\newtheorem{example}{Example}[section]
\newtheorem{examples}{Example}[subsection]

\newtheorem{remark}{Remark}[section]
\newtheorem{remarks}{Remarks}[section]
\newtheorem{definition}{Definition}[section]

\numberwithin{equation}{section} 

\DeclareMathOperator{\End}{End}

\DeclareMathOperator{\sgn}{sgn}

\DeclareMathOperator{\GL}{GL}

\DeclareMathOperator{\SO}{SO}

\def\ra{{\rightarrow}}

\def\det{\mathrm {det}}

\def\Pf{\mathrm {Pf}}

\def\End{\mathrm {End}}

\def\span{\mathrm {span}}

\def\span{\mathrm {span}}

\def\Gr{\mathrm {Gr}}

\def\&{&{\hskip -20pt}}

\DeclarePairedDelimiter{\no}{:}{:}

\def\be{\begin{equation}}
\def\ee{\end{equation}}

\def\bea{\begin{eqnarray}}
\def\eea{\end{eqnarray}}

\def\bt{\begin{theorem}}
\def\et{\end{theorem}}

\def\bex{\begin{example}\small \rm}
\def\eex{\end{example}}

\def\bexs{\begin{examples}\small \rm}
\def\eexs{\end{examples}}

\def\ra{\rightarrow}
\def\ss{\subset}

\def\br{\begin{remark}\small \rm}
\def\er{\end{remark}}

\def\CC{{\mathcal C}}

\def\FF {{\mathcal F}}

\def\HH{{\mathcal H}}
\def\II {{\mathcal I}}

\def\PP{{\mathcal P}}
\def\QQ {{\mathcal Q}}

\def\Bb{\mathbf{B}}

\def\Ib{\mathbf{I}}

\def\Nb{\mathbf{N}}

\def\Pb{\mathbf{P}}

\def\Zb{\mathbf{Z}}

\def\pb{\mathbf{p}}

\def\rb{\mathbf{r}}

\def\tb{\mathbf{t}}

\def\xb{\mathbf{x}}

 \def\grl{\mathfrak{l}}

 \def\gro{\mathfrak{o}} 
\def\grP{\mathfrak{P}}

 \def\grgl{\mathfrak{gl}}

\newcommand{\Pa}{\mathop\mathrm{P}\nolimits}

\def\bp{\begin{Proposition}\rm}
\def\ep{\end{Proposition}}
\def\bc{\begin{corollary}}
\def\ec{\end{corollary}}
\def\bl{\begin{lemma}\em}
\def\el{\end{lemma}}
\def\be{\begin{equation}}
\def\ee{\end{equation}}
\def\br{\begin{remark}\rm\small}
\def\er{\end{remark}}
\def\brs{\begin{remarks}.\\ \rm\
\begin{enumerate}}
\def\ers{\end{enumerate}\end{remarks}}
\def\bea{\begin{eqnarray}}
\def\eea{\end{eqnarray}}

\begin{document}

\begin{center}
\begin{Large}\fontfamily{cmss}
\fontsize{17pt}{27pt}
\selectfont
	\textbf{Bilinear expansions of lattices of KP $\tau$-functions in BKP $\tau$-functions: a fermionic approach}
	\end{Large}
	
\bigskip \bigskip
\begin{large} 
J. Harnad$^{1, 2}$\footnote[1]{e-mail:harnad@crm.umontreal.ca} 
and A. Yu. Orlov$^{3}$\footnote[2]{e-mail:orlovs55@mail.ru} 
 \end{large}
 \\
\bigskip

\begin{small}
$^{1}${\em Centre de recherches math\'ematiques, Universit\'e de Montr\'eal, \\C.~P.~6128, succ. centre ville, Montr\'eal, QC H3C 3J7  Canada}\\
$^{2}${\em Department of Mathematics and Statistics, Concordia University\\ 1455 de Maisonneuve Blvd.~W.~Montreal, QC H3G 1M8  Canada}\\
$^{3}${\em Shirshov Institute of Oceanology, Russian Academy of Science, Nahimovskii Prospekt 36, Moscow 117997, Russia }
\end{small}
 \end{center}

\medskip

\begin{abstract}
We derive a bilinear expansion expressing  elements of a lattice of Kadomtsev-Petvishvili (KP) $\tau$-functions,
 labelled by partitions, as a sum over products of pairs of elements of an associated lattice of BKP $\tau$-functions, 
labelled by strict partitions. This generalizes earlier results relating determinants and Pfaffians of minors 
of  skew symmetric matrices, with applications to Schur functions and Schur $Q$-functions.
  It is deduced using the representations of  KP and BKP $\tau$-functions as vacuum expectation values (VEV's) 
  of products of fermionic  operators of charged and neutral type, respectively.
  The lattice is generated by insertion of products of pairs of charged creation and annihilation operators.  
 The result follows from expanding the product as a sum of monomials  in the neutral fermionic generators
   and applying a factorization theorem for  VEV's of products of operators in the mutually commuting subalgebras.
   Applications include the case of inhomogeneous polynomial $\tau$-functions of KP and BKP type.
     \end{abstract}
\bigskip

\section{Introduction}
\label{introduction}

In \cite{BHH}, a bilinear expansion was derived, expressing determinants of submatrices of a skew symmetric matrix 
as sums over products of  Pfaffians of their principal minors. This was shown to follow from the 
relation between the Pl\"ucker \cite{GH} and Cartan \cite{Ca} embeddings 
of Grassmannians and orthogonally isotropic Grassmannians in the projectivization of an exterior space.  
Using fermionic methods, a function theoretic realization of this identity was derived in \cite{HO}, expressing Schur functions
$s_\lambda$, labelled by integer partitions $\lambda$, as sums over products of pairs $Q_{\mu^+}Q_{\mu^-}$
of Schur Q-functions, labelled by pairs $(\mu^+, \mu^-)$ of {\em strict} partitions \cite{Mac}. Schur functions and Schur $Q$-functions  are the basic building blocks for $\tau$-functions of the KP and BKP integrable hierarchies, and the simplest case of this bilinear relation  is the well-known fact that the square of any BKP $\tau$-function is the restriction of  a KP $\tau$-function to the odd flow variables \cite{DJKM1, JM, vdLO}.

The integrable KP  and BKP hierarchies are closely related to the $A_\infty$ and  $B_\infty$  root lattices,
respectively \cite{JM}. In this work  we define, for every $\GL(\infty)$ group element $\hat{g}$, an $A_\infty$ lattice of KP $\tau$-functions,
 labelled by partitions, and for every $\SO(\infty)$ group element $\hat{h}$,  a $B_\infty$ lattice of BKP $\tau$-functions,
  labelled by strict partitions.   When $\hat{g}$ and $\hat{h}$ are suitably related (eq.~(\ref{g_h_def})), we show (Theorem \ref{KP_BKP_lattice_bilinear}) how a generalization of the bilinear expansions  of  \cite{BHH} and \cite{HO} follows,
 expressing the elements of the lattice of KP $\tau$-functions as sums over products of  pairs of  $\tau$-functions from the 
 associated lattice of BKP type.  Varying the choice of infinite orthogonal group element  
  $\hat{h}$ determining the BKP $\tau$-function at the origin of the lattice  gives all possible functional realizations of 
  such bilinear relations.  Two key ingredients underlying this result are:    the identification of {\em  polarizations} 
  (Section \ref{polarizations}) associated  with a given partition, which are just 
  pairs $(\mu^+, \mu^-)$  of strict partitions sharing the same intersections and unions,  together with the {\em factorization Lemma} \ref{factorization_lemma},   which allows us to express vacuum expectation values (VEV's) of products of mutually anticommuting
   sets of neutral fermionic operators $\{\phi^\pm_j\}_{j \in \Zb}$ as products of VEV's of the two separate types.
  The results  of \cite{BHH} and \cite{HO} are recovered as special cases, together with other examples involving 
more general lattices of KP and BKP $\tau$-functions,  in particular, those consisting of inhomogeneous 
polynomial $\tau$-functions. 


\section{Fermionic Fock space, creation and annihilation operators, KP and BKP $\tau$-functions}
\label{fermi_ops}

\subsection{Charged fermions and KP flows}
\label{charge_fermion_KP}

We recall some definitions and notations for free fermions and KP flows \cite{JM, MJD, HB}. 
Given a separable Hilbert space $\HH$ with orthonormal basis $\{e_j\}_{j\in \Zb}$
and its dual space $\HH^*$, with dual basis $\{e^j\}_{j\in \Zb}$ satisfying
\be
e^j(e_k) = \delta_{jk},
\label{basis_duality}
\ee
the fermionic Fock space $\FF$ is defined as the $\Zb$-graded semi-infinite wedge product 
\be
\FF:= \Lambda^{\infty/2}\HH =\sum_{n\in \Zb} \FF_n, 
\label{fock_def}
\ee
where $\FF_n$ denotes the $n$th  fermionic charge sector.
The elements of $\FF$ are denoted by {\em ket} vectors $| \cdot \rangle$, while those of the dual space $\FF^*$
are denoted by {\em bra} vectors $\langle \cdot |$. Let 
\be
\lambda = (\lambda_1, \lambda_2, \dots, \lambda_{\ell(\lambda)}, 0, 0 \dots)
\label{part_lambda}
\ee
be an extended integer partition of {\em length} $\ell(\lambda)$, consisting of a weakly decreasing
sequence $\lambda_1 \ge \cdots \ge \lambda_{\ell(\lambda})> 0$ of $\ell(\lambda)$ positive integers, 
followed by an infinite sequence  \hbox{$(\{\lambda_{\ell(\lambda)+j}=0\}_{j=1, 2, \dots}$}of $0$'s.
 An orthonormal basis $\{|\lambda; n\rangle\}$ for $\FF_n$ is provided by the elements 
\be
|\lambda; n\rangle\:= e_{l_1} \wedge e_{l_2} \wedge \cdots, 
\label{fock_basis}
\ee
where $(l_1, l_2, \dots)$ is a strictly decreasing sequence of integers, called {\em particle positions}.
  defined in terms of the partition $\lambda$ by
  \be
  l_j = \lambda_j -j +n, \quad j=1, 2, \dots.
  \label{part_positions}
  \ee
which, after the first $\ell(\lambda)$ elements, saturate to the  sequence of all subsequent decreasing integers,
$n-\ell(\lambda), n-\ell(\lambda) -1, \dots$. The dual basis vectors are denoted $\langle \lambda; n|$, and
satisfy
\be
\langle \lambda; n | \mu; m\rangle = \delta_{\lambda, \mu} \delta_{nm}.
\label{fock_basis_duality}
\ee
The vacuum states  in each sector $\FF_n$  correspond to the trivial partition
$\emptyset :=(0, 0, \dots)$ and are denoted
\be
| n\rangle := |\emptyset; n\rangle.
\label{n_vac}
\ee

The irreducible representation 
\bea
\Gamma: \CC_{\HH+\HH^*}(Q)  &\& \ra \End(\FF) \cr
\Gamma: \xi &\& \mapsto \Gamma_\xi
\label{clifford_Gamma}
\eea
of the Clifford algebra $\CC_{\HH+\HH^*}(Q)$ algebra 
on $\HH + \HH^*$ with respect to the scalar product
\be
Q(v + \mu, w + \nu) := \nu(v) + \mu(w), \quad v,w \in \HH, \ \mu, \nu \in \HH^*
\label{Q_scalar_prod}
\ee
 is generated by scalars and linear elements, the latter acting by exterior and interior multiplication:
 \be
 \Gamma_v := v\wedge, \quad \Gamma_\mu := i_\mu.
 \label{clifford_rep_gen}
 \ee
The fermionic creation and annihilation operator are representations of the basis elements:
\be
\psi_j := \Gamma_{e_j} = e_j \wedge, \quad \psi^\dag := \Gamma_{e^j} := i_{e^j}
\ee
and satisfy the anticommutation relations
\be
[\psi_j, \psi_k]_+ = 0, \quad [\psi^\dag_j, \psi^\dag_k]_+ = 0, \quad [\psi_j, \psi^\dag_k]_+ = \delta_{jk}.
\label{psi_i_anticomm}
\ee

For $j>0$, $\psi_{-j}$ and $\psi^\dag_{j-1}$ (resp. $\psi^\dag_{-j}$ and $\psi_{j-1}$) 
annihilate the right (resp. left) vacua:
\bea
\psi_{-j} |0\rangle &\&=0, \quad \psi^\dag_{j-1} |0\rangle =0, \quad \forall  \, j >0,
\label{vac_annihil_psi_j_r}
 \\
\langle 0| \psi^\dag_{-j}  &\&=0, \quad \langle 0 | \psi_{j-1}  =0, \quad \forall  \, j >0.
\label{vac_annihil_psi_j_l}
\eea
The basis states may be expressed by acting with a product of pairs
of operators $\psi_j\psi^\dag_k$ on the vacuum:
\be
|\lambda; 0\rangle =: |\lambda \rangle =  (-1)^{\sum_{j=1}^r\beta_j} \prod_{j=1}^r 
 \left( \psi_{\alpha_j} \psi_{-\beta_j-1}^\dag \right) |0\rangle ,
 \label{lambda_frob_vac}
\ee
where $(\alpha_1, \dots, \alpha_r | \beta_1, \dots, \beta_r)$ are the Frobenius  indices \cite{Mac}  of the  partition $\lambda$.

The fermionic representation of elements $g\in \GL_0(\HH)$ of the connected component of the 
group $\GL(\HH)$ can be expressed as exponentials of $\grgl(\HH)$ algebra elements, which are linear
 combinations of normal ordered products $\no{\psi_j \psi^\dag_k}$ of creation and annihilation operators \cite{JM}
\be
\hat{g}(\tilde{A}) = e^{\sum_{i,j \in \Zb} \tilde{A}_{jk} \no{\psi_j \psi^\dag_k}},
\label{hat_g_exp_A}
\ee
where normal ordering 
\be
\no{\hat{L}_1 \hat{L}_2} = \hat{L}_1 \hat{L}_2 - \langle 0 | \hat{L}_1 \hat{L}_2 | 0 \rangle.
\label{normal_order}
\ee 
of a product of linear elements $\hat{L}_1 \hat{L}_2$ is defined so the vacuum expectation value (VEV) vanishes.
There are various conventions on restricting the permissible values of the doubly infinite matrix  elements $\tilde{A}_{jk}$,
depending on the class of solutions to the KP hierarchy one wishes to include \cite{Sa, JM, SW}.

The KP flow parameters, denoted
\be
{\bf t}  =(t_1, t_2, \dots), 
\label{KP_flow_params}
\ee 
may be interpreted as linear exponential coordinates on an infinite abelian subgroup $\Gamma_+ \ss \GL(\HH)$
defined by
 \bea
 \Gamma_+=\{\gamma_+({\bf t}) := e^{\sum_{j=1}^\infty t_j \Lambda^j} \}, 
 \label{Gamma_+_def}
 \eea
 where $\Lambda\ \in \End(\HH)$ is the shift operator which, acts on basis elements by
 \be
 \Lambda(e_j) = e_{j-1}.
 \label{shift_Lambda}
 \ee
 The KP flows are generated by the action of $\Gamma_+$ on the infinite Grassmannian
 $\Gr_{\HH_+}(\HH)$ of subspaces $w\ss \HH$ commensurate with the subspace
 \be
 \HH_+ = \span\{e_{-j}\}_{j\in \Nb^+},
 \label{HH_+_def}
 \ee
 satisfying certain admissibility conditions \cite{SW, Sa}. They are mapped into the projectivization
 $\Pb(\FF)$ of $\FF$ by the Pl\"ucker map:
 \be
 \grP\grl: \span\{w_j\}_{j\in \Nb^+} \ra [w_1 \wedge  w_2 \wedge \cdots ], 
 \label{pluskcer_map}
\ee
where $\{w_j\}_{j\in \Nb^+} $ is any admissible basis of $w\in\Gr_{\HH_+}(\HH)$.
 The image $\grP\grl(\HH_+)$ of $\HH_+\ss \HH$ is the projectivized vacuum state $[|0\rangle]$,
 and the Pl\"ucker map intertwines the action of $\Gamma_+$ induced on $\Gr_{\HH_+}(\HH)$
 with the (projectivized) fermionic representation of   $\Gamma_+ \ss \GL(\HH)$ on $\FF$,
 given by
\be
\Gamma_{\gamma_+({\bf t})}=  \hat{\gamma}_+({\bf t}) := e^{\sum_{j=1}^\infty t_j J_j},
\label{hat_gamma_+}
\ee
where 
\be
J_j := \sum_{k\in \Zb} \psi_k \psi^\dag_{k+j}, \quad j=1,2, \dots
\label{J_j_def}
\ee
are the positive Fourier components of the current operator  \cite{JM}, which mutually commute
\be
[J_n, J_m ] =0,\quad n,m \in \Nb^+.
\label{current_comp_cr}
\ee
By (\ref{vac_annihil_psi_j_r}), the  current components $J_n$  annihilate the vacuum $|0\rangle$:
\be
J_n |0 \rangle =0, \quad n\in \Nb^+
\label{J_n_vac_0}
\ee
and hence $\Gamma_+$ stabilizes the vacuum
\be
\hat{\gamma}_+({\bf t}) |0\rangle = |0\rangle.
\label{gamma_+_vac}
\ee

For each $g\in \GL_0(\HH)$, we define the corresponding KP $\tau$-function
\be
\tau^{KP}_g({\bf t}) = \langle 0 | \hat{\gamma}_+({\bf t})\hat{g} | 0 \rangle ,
\label{KP_tau_fermi}
\ee
which satisfies the Hirota bilinear residue equations \cite{Sa, JM} of the KP hierarchy.
More generally, $\hat{g}$ need not be a  $\GL_0(\HH)$ group element,
 but any element of the Clifford algebra\index{Clifford algebra}  satisfying the 
 bilinear commutation relation
\be
\left[ \II ,  \hat{g} \otimes \hat{g} \right] = 0, 
\label{bilinear_commut}
\ee
where $\II$ is the endomorphism of  $\FF \otimes \FF$ defined by
\be
\II:= \sum_{j\in \Zb} \psi_j \otimes \psi_j^\dag.
\label{I_bilinear}
\ee

\subsection{Neutral fermions and BKP flows}
\label{neutral_fermionc_BKP}

 Neutral fermions $\phi^+_j$ and $\phi^-_j$ are defined \cite{DJKM1, DJKM2, JM, You, KvdL1} by 
\be
\phi^+_j :=\frac{\psi_j+ (-1)^j\psi^\dag_{-j}}{\sqrt 2},\quad 
\phi^-_j :=i\frac{\psi_j-(-1)^j \psi^\dag_{-j}}{\sqrt 2},\quad j \in \Zb
\label{neutral_charged}
\ee
(where $i=\sqrt{-1}$), and satisfy
\be
 [\phi^+_j,\phi^-_k]_+=0,\quad [\phi^+_j,\phi^+_k]_+ = [\phi^-_j,\phi^-_k]_+ =(-1)^j \delta_{j+k,0}.
 \label{neutral_ccr}
\ee
In particular,
\be
(\phi^+_0)^2=(\phi^-_0)^2=\tfrac{1}{2}.
\label{square_phi_0}
\ee
Acting on the vacua $|0\rangle$ and $|1 \rangle$, we have
\bea
\phi^+_{-j}|0\rangle &\&= \phi^-_{-j} |0\rangle =\phi^+_{-j}|1\rangle = \phi^-_{-j} |1\rangle =0 , \quad  \forall  j > 0  , \quad \forall  j > 0, 
\label{phi_pm_j_01vac_r} \\
\langle 0| \phi^+_{j} &\&= \langle 0|\phi^-_{j}=  \langle 1| \phi^+_{j} = \langle 1|\phi^-_{j}  =0 , \quad  \forall  j > 0 ,
\label{phi_pm_j_01vac_l}
\\
\phi^+_0|0\rangle &\& = 
- i \phi^-_0 |0\rangle =
\tfrac{1}{\sqrt{2}} \psi_0|0\rangle = \tfrac{1}{\sqrt{2}} |1\rangle  ,
\label{phi_pm_0vac_r}    
\\
 \langle 0| \phi^+_0 &\& = 
 i \langle 0|\phi^-_0  =
\tfrac{1}{\sqrt{2}} \langle  0 | \psi_0^\dag = \tfrac{1}{\sqrt{2}}\langle 1|.
\label{phi_pm_0vac_l}  
\eea

Let $\tilde{\HH}^\pm \ss \HH + \HH^*$ be the two mutually orthogonal subspaces  defined
by
\bea
\tilde{\HH}^+  &\&:= \span\{f_j^+:=\frac{1}{\sqrt{2}} ( e_j + (-1)^{-j} e^{-j}\}_{j\in \Zb}, \cr
 \cr
\tilde{\HH}^-&\& := \span\{f_j^-:= \frac{i}{\sqrt{2}} ( e_j - (-1)^{-j} e^{-j}\}_{j\in \Zb} .
\eea
Restricting  the scalar  product $Q$ to each of these subspaces, we have the induced scalar products $Q_\pm$ on $\HH^\pm$
defined by
\bea
Q_+(f^+_j, f^+_k) &\&:=  Q(f^+_j, f^+_k) = (-1)^j\delta_{j+k, 0}, \cr
&\& \cr
 Q_-(f^-_j, f^-_k) &\&:= Q(f^-_j, f^-_k) = (-1)^j\delta_{j+k, 0}. 
 \label{Q_pm_def}
\eea
We thus have two mutually commuting fermionic representations of the orthogonal algebras $\gro(\tilde{\HH}^\pm,Q_\pm)$
defined by the span of the bilinear elements $\{\phi^\pm_j\phi^\pm_k\}_{j,k \in \Zb}$. 

The fermionic representations of elements $h^\pm \in \SO(\tilde{\HH}^\pm, Q_\pm)$ in the connected component of the corresponding 
mutually commuting orthogonal  groups $\SO(\tilde{\HH}^\pm, Q_\pm)$, are defined by exponentiation \cite{DJKM1, DJKM2, vdLO}
\be
\hat{h}^\pm(A) := e^{\tfrac{1}{2}\sum_{j,k \in \Zb}  A_{jk} \no{\phi^\pm_j \phi^\pm_{k}} },
\label{hat_h_alpha_pm_exp_A}
\ee
 where $\{A_{jk}\}_{j,k \in \Zb}$ are the elements of a doubly infinite skew symmetric matrix $A$.
Let $\FF_{\phi^\pm}\ss \FF$ be the two subspaces of $\FF$ defined as
\be
\FF_ {\phi^+}:= \span\{|\hat{\alpha}^+ \big)\}, \quad 
\FF_ {\phi^-}:= \span\{|\hat{\alpha}^- \big)\}
 \label{FF_pm_def}
\ee
where
\be
|\alpha^+ \big) := \phi^+_{\alpha_1} \cdots \phi^+_{\alpha_r}|0 \rangle, \quad
|\alpha^- \big) := \phi^-_{\alpha_1} \cdots \phi^-_{\alpha_r}|0 \rangle.
 \label{alpha_pm}
\ee
Then $\FF_ {\phi^+}$ and $\FF_ {\phi^-}$ are invariant under the fermionic representations
(\ref{hat_h_alpha_pm_exp_A}) of the groups $\SO(\tilde{\HH}^\pm, Q_\pm)$.

The positive neutral fermion current components $J^{B+}_j$ and $\hat{J}^{B-}_j$ are defined as
\be
J^{B+}_j:=\tfrac 12 \sum_{k\in\mathbb{Z}} (-1)^{k+1}\phi^+_k \phi^+_{-k-j},\quad 
J^{B-}_j=\tfrac 12 \sum_{k\in\mathbb{Z}} (-1)^{k+1}\phi^-_k\phi^-_{-k-j}, \quad j\in \Nb^+,
\label{J_j_B_def}
\ee
from which it follows that the even components $\{J^{B\pm}_{2j}\}$ vanish, while the odd ones mutually commute,
\be
[J^{B+}_{2j-1}, J^{B+}_{2k-1}]  = 0,\quad [J^{B-}_{2j-1}, J^{B-}_{2k-1}]=0,\quad
[J^{B+}_{2j-1}, J^{B-}_{2k-1}] =0, \quad j,k \in \Nb^+ .
\label{JB_pm_crs}
\ee

It also follows, from  (\ref{phi_pm_j_01vac_r}), that the neutral current components $J^{B\pm}_j$  
 annihilate the vacua $|0\rangle$ and $|1\rangle$,
\be
J^{B+}_j|0 \rangle = 0, \quad J^{B-}_j|0 \rangle =0, \quad J^{B+}_j|1 \rangle = 0, \quad J^{B-}_j|1 \rangle, \quad j \in \Nb^+, 
\label{J_B_vac}
\ee
and, from (\ref{J_j_def}),  (\ref{J_j_B_def}),  that:
\bl
For  odd $n= 2j-1$, we have
\be
 J_{2p-1}=J^{B+}_{2j-1}+J^{B-}_{2j-1}, \quad j \in \Nb^+.
 \label{J_J_Bpm}
\ee
\el

The  BKP flow parameters are denoted
\be
{\bf t}_B:=(t_1, t_3, \dots ),
\label{BKP_flow_params}
\ee
and these may be viewed as determining a subset  $\{\tb'\} $ of the $\tb$'s , where
\be
\tb':=(t_1,0,t_3,0,t_5,\dots).
\label{t_prime_def}
\ee
Following \cite{DJKM1, DJKM2,  You, JM, vdLO}, we  define  two mutually commuting abelian groups of BKP flows 
\hbox{$\Gamma^{B+} = \{\gamma^{B+}({\tb}_{B+})\}$} and 
\hbox{$\Gamma^{B-} = \{\gamma^{B-}({\tb}_B)\}$}, 
with Clifford representations
\be
\hat{\gamma}^{B+}({\tb}_B) :=  e^{\sum_{j=0}^\infty J^{B+}_{2j-1} t_{2j-1} } ,
 \quad \hat{\gamma}^{B-}({\tb}_B) := e^{\sum_{j=1}^\infty J^{B-}_{2j-1} t_{2j-1} }.
 \label{hat_gamma_B_pm}
\ee
By (\ref{J_B_vac}), these  both stabilize the vacua $|0\rangle$ and $| 1\rangle$
\bea
\hat{\gamma}^{B+}({\tb}_B)|0\rangle &\&= |0\rangle, \quad \hat{\gamma}^{B+}({\tb}_B)|1\rangle = |1\rangle, \cr
\hat{\gamma}^{B-}({\tb}_B)|0\rangle &\&= |0\rangle, \quad \hat{\gamma}^{B-}({\tb}_B)|1\rangle = |1\rangle.
\label{J_sum_J+-}
\eea
It also follows that
\be
\hat{\gamma}_+({\bf t}') = \hat{\gamma}^{B+}_+({\bf t}_B)  \hat{\gamma}^{B-}_+({\bf t}_B) .
\label{gamma_factoriz}
\ee
The BKP $\tau$-function, defined as
\be
\tau^{BKP}_h({\bf t}_B) = \langle 0| \hat{\gamma}^{B+}({\bf t}_B) \hat{h}^+ | 0 \rangle = \langle 0| \hat{\gamma}^{B-}({\bf t}_B) \hat{h}^- | 0 \rangle,
\label{BKP_tau_fermi}
\ee
then satisfies the BKP Hirota bilinear residue equations \cite{DJKM1, DJKM2, JM}.

More generally, $\hat{h^\pm}$ need not be $\SO(\tilde{\HH}^\pm)$ group elements for
(\ref{BKP_tau_fermi}) to define a BKP $\tau$-function.
Define $\II_B^\pm\in \End(\FF_{\phi^\pm} \otimes \FF_{\phi^\pm})$ as
\be
\II_B^\pm := \sum_{j\in \Zb}(-1)^j \phi^\pm_j  \otimes \phi^\pm_{-j} ,
\label{II_B}
\ee
and let $P^\pm \in  \End(\FF_{\phi}^{\pm})$ be  defined on the bases $\{|\alpha^\pm)\}$ as
\be
P^{\pm}|\alpha^\pm) := (-1)^{m(\alpha^\pm)}|\alpha^\pm), 
\label{FF_partity_op}
\ee
where $m(\alpha^\pm)$ is the length (i.e.~the cardinality) of $\alpha^\pm$.
Then  \ref{BKP_tau_fermi}) defines a BKP $\tau$-function \cite{HB, DJKM1, DJKM2, JM} provided the bilinear relations 
  \be
\II_B^\pm  | v^\pm) \otimes |v^\pm ) =   (\phi_0 \otimes \phi_0)P^\pm | v^\pm ) \otimes  P^\pm| v^\pm ).
\label{neutral_fermion_bilinear}
\ee
are satisfied by the elements 
\be
|v^\pm) = \hat{h}^\pm|0\rangle.
\label{decomp_FF_N_pm)el}
\ee

We also have the following key factorization lemma \cite{HO, You, JM}.
 \begin{lemma}[\bf Factorization]
 \label{factorization_lemma}
 If $U^+$ and $U^-$ are either even or odd degree elements of the subalgebra generated by the operators
 $\{\phi^+_i\}_{i \in \Zb}$ and  $\{\phi^-_i\}_{i \in \Zb}$ respectively, the VEV of their product can be factorized as:
\be
\langle 0 | U^+ U^-|0\rangle =
\begin{cases}
\langle 0 |  U^+|0\rangle \langle 0| U^- |0\rangle 
&  
\text{ if } U^+ \text{ and } U^-\text{ are both of even degree}\\
0  &  \text{ if } U^+ \text{ and }  U^- \text{ have different  parity }\\ 
2i \langle 0|U^+\phi^+_0|0\rangle \langle 0|U^-\phi^-_0 |0\rangle &  
\text{  if }  U^+ \text{ and }  U^+ \text{ are both of odd degree}.
\end{cases}
\label{gen_factorization}
\ee
\end{lemma} 

Defining
\be
g(h(A)) = h^+(A)h^-(A), \quad \hat{g}(h(A)) :=\hat{h}^+(A)\hat{h}^-(A)
\label{g_h_def}
\ee
and using the relations (\ref{neutral_charged}) between neutral and charged fermions and the mutual commutativity
of the algebras $\gro(\tilde{\HH}^\pm,  Q_\pm)$, we can re-express $\hat{g}$ as the fermionic representation
of an element of $\GL(\HH)$
\be
\hat{g}(\tilde{A})=\hat{g}(h(A)) = e^{\sum_{j,k \in \Zb} (-1)^kA_{j, -k} \no{\psi_j \psi^\dag_k}}
\label{fermi_rep_g_h}
\ee
where
\be
\tilde{A} = (-1)^k A_{j, -k}
\label{tilde_A_A_skew}
\ee
 Comparing formulae (\ref{KP_tau_fermi}) and  (\ref{BKP_tau_fermi}) for the KP $\tau$-function
 $\tau^{KP}_{g(h)}({\bf t}')$ and the BKP $\tau$-function  $\tau^{BKP}_h({\bf t}_B)$, using
(\ref{gamma_factoriz}) and (\ref{g_h_def}) and Lemma \ref{factorization_lemma}, we obtain the
following well-known result \cite{DJKM1,  DJKM2}
\be
\tau^{KP}_{g(h)}({\bf t}') = \left(\tau^{BKP}_h({\bf t}_B) \right)^2.
\label{tau_KP_g_h}
\ee

The next section is aimed at obtaining the general bilinear relation expressing  a lattice of KP $\tau$-functions,  
labelled by partitions, associated with  $\tau^{KP}_{g(h)}({\bf t}')$,  as sums over products of the elements  
of a lattice of BKP $\tau$-functions, lablled by strict partitions, similarly associated with $\tau^{BKP}_h({\bf t}_B)$

\section{Bilinear expansions of KP $\tau$-function lattices in BKP $\tau$-functions}
\label{ bilnear_KP_BKP}
In this section, we define lattices of KP and BKP $\tau$-functions related to a given pair 
\hbox{$\tau^{KP}_{g(h)}(\tb)$, $\tau^{BKP}_h({\bf t}_B)$}, labelled by partitions and strict partitions, 
respectively, and derive a bilinear expansion for the first of these in terms of the second.


\subsection{Lattices of KP and BKP $\tau$-functions}
\label{KP_BKP_tau_lattice}
In the following, we view the KP and BKP flow variables ${\bf t}$ and ${\bf t}_B$, 
which are the arguments of our $\tau$-functions, as independent coordinates.
To make contact with the usual definitions of Schur functions and Schur $Q$-functions,
as symmetric polynomials, or symmetric functions of a finite or infinite set of variable $\{x_j\}_{j\in\Nb^+}$,
we recall that the variables ${\bf t}$ and ${\bf t}_B$ may be restricted to equal the (normalized) power sums
in these, 
\be
t_i = {1\over i} \sum_{j=1}^N{x_j^i},
\ee
(where N could be a finite positive integer or $\infty$).
We then denote the infinite sequence ${\bf t}=(t_1, t_2 \dots )$ as
\be
{\bf t}=  [{\bf x}] := ([{\bf x}]_1, [{\bf x}]_2, \dots), \quad [{\bf x}]_i :=   {1\over i} \sum_{j=1}^N{x_j^i}, \quad i \in \Nb^+.
\label{t_norm_power_sum}
\ee
However, when we specialize to the values $\tb = \tb'$ , in which all the even flow variables vanish,
we cannot interpret these as the usual power sums (\ref{t_norm_power_sum}). Instead, in order to
restrict these to symmetric polynomials, we use the {\em supersymmetric} parametrization \cite{Iv}
\be
\tb' = [{\bf y}] - [-{\bf y}], \quad t_{2i-1} = {2\over 2i-1} \sum_{j =1}^Ny_j^{2i-1},
\label{t_prime_y_sum}
\ee
where ${\bf y} =(y_1, y_2, \dots)$ is an auxiliary (finite or infinite) sequence of independent variables.

For any KP $\tau$-function $\tau^{KP}_g(\tb)$, we associate an $A_\infty$ lattice of KP $\tau$-functions
\be
\pi_{(\alpha|\beta)}(g) ({\bf t}):= \langle 0 | \hat{\gamma}_+({\bf t}) \hat{g} | \lambda \rangle,
\label{pi_alpha_beta}
\ee
labelled by partitions $\lambda =(\alpha|\beta)$  and, for any   BKP $\tau$-function $\tau^{BKP}_h(\tb_B)$,
a $B_\infty$ lattice of BKP $\tau$-functions
\be
\kappa_{\alpha} (h)({\bf t}_B) := \langle 0 | \hat{\gamma}^{B+}({\bf t}_B) \hat{h}^+ | \alpha^+ \big) =
 \langle 0 | \hat{\gamma}^{B-}({\bf t}_B) \hat{h}^- | \alpha^- \big),
 \label{kappa_alpha}
\ee
labelled by strict partitions $\alpha$ of even cardinality $r$.
\begin{remark}
Note that if we change the elements $\hat{g} \ra \hat{g} \hat{g}_0 $ and 
$\hat{h}^\pm \ra \hat{h} \hat{h}_0 $ in (\ref{pi_alpha_beta}) and (\ref{kappa_alpha})  by right multiplication by elements 
$\hat{g}_0$,  $\hat{h}_0$ of the stabilizer of the vacuum $|0\rangle$, this has no effect
on the $\tau$-functions $\tau^{KP}_g(\tb)$ and $\tau^{BKP}_h(\tb_B)$, but
will change the other elements $\pi_{(\alpha|\beta)}(g)(\tb)$, $\kappa_{\alpha} (h)({\bf t}_B) $
of the associated lattices.

\end{remark}

By Wick's theorem, $\pi_{(\alpha|\beta)}(g) ({\bf t})$ is the determinant of the matrix 
$\pi_{(\alpha_j|\beta_k)}(g)({\bf t}) $ corresponding to hook partitions $(\alpha_j|\beta_k)$ 
\be
\pi_{(\alpha|\beta)}(g)({\bf t}) = \det\left(\pi_{(\alpha_j|\beta_k)}(g)({\bf t}) \right)_{1\le j, k\le r}, 
\label{pi_alpha_beta_det}
\ee
and  $\kappa_{\alpha} (h)({\bf t}_B)$ is the Pfaffian of the skew matrix $\kappa_{(\alpha_j, \alpha_k)} (h)({\bf t}_B)$
corresponding to the Frobenius rank $r=2$ strict partitions $\alpha = (\alpha_j, \alpha_k)$:
\be
\kappa_{\alpha} (h)({\bf t}_B)  = \Pf\left( \kappa_{(\alpha_j, \alpha_k)} (h)({\bf t}_B)\right)_{1\le j, k\le r}.
\label{kappa_alpha_pfaff}
\ee
where, for $\alpha_k\ge\alpha_j$, $\kappa_{(\alpha_j, \alpha_k)(h) ({\bf t}_B)} $ is  defined  by
\be
\kappa_{(\alpha_j, \alpha_k)}(h)({\bf t}_B)  := - \kappa_{(\alpha_k, \alpha_j)}(h)({\bf t}_B) .
\label{kappa_alpha_def}
\ee

\subsection{Polarizations}
\label{polarizations}

\begin{definition}
A {\em polarization} of $(\alpha|\beta)$,  is a pair  $(\mu^+, \mu^-)$
of strict partitions with cardinalities  (or {\em lengths})
\be
m(\mu^+):= \#(\mu^+), \quad  m(\mu^-)):= \#(\mu^-)
\label{m_pm_def}
\ee
(including  possibly a zero part $\mu^+_{m(\mu^+)}=0$ or $\mu^-_{m(\mu^-)}=0$), satisfying
\be
\mu^+ \cap \mu^- = \alpha \cap I(\beta), \quad \mu^+ \cup \mu^- = \alpha \cup I(\beta),
\label{mu_pm_cap_cup}
\ee
where
\be
I(\beta) := (I_1(\beta), \dots I_r(\beta))
\label{I_beta}
\ee
is the strict partition \cite{Mac} with parts
\be
I_j(\beta) = \beta_j +1, \quad j=1, \dots r.
\label{I_j_beta}
\ee
The set of all polarizations of $(\alpha|\beta)$  is denoted $\PP(\alpha,\beta)$.
\end{definition}
We denote the strict partition obtained by intersecting  $\alpha$ with $I(\beta)$ as
\be
S := \alpha \cap I(\beta),
\label{S_def}
\ee
and its cardinality as
\be
s := \#(S).
\ee
Since both $\alpha$ and $I(\beta)$ have cardinality $r$, it follows that
\be
m(\mu^+) + m(\mu^-) = 2r,
\label{m_pm_sum}
\ee
so $m(\mu^\pm)$ must have the same parity.   It is easily verified \cite{HO} that the cardinality of $\PP(\alpha, \beta)$ is $2^{2r -2s}$.
The following was proved in \cite{HO}.
\begin{lemma}[\bf Binary sequence associated to a polarization]
For every polarization $\mu:=(\mu^+, \mu^-)$ of $\lambda=(\alpha| \beta)$ there is a unique binary sequence
of length $2r$
\be 
\epsilon(\mu) =(\epsilon_1(\mu), \dots, \epsilon_{2r}(\mu)),
\label{epsilon_mu_binary}
\ee
with
\be
\epsilon_j(\mu) =\pm, \quad j=1, \dots 2r,
\label{epsilon_j_binary}
\ee
such that 
\begin{enumerate}
\item The sequence of pairs
\be
((\alpha_1, \epsilon_1(\mu)), \cdots  (\alpha_r, \epsilon_r(\mu)), (\beta_1+1, \epsilon_{r+1}(\mu)), \dots , ( \beta_r+1, \epsilon_{2r}(\mu))
\label{alpha_beta_sequence}
\ee
is a permutation of the sequence
\be
((\mu^+_1,+), \cdots ( \mu^+_{m^+(\mu)},+), (\mu^-_1, -),  \dots  , (\mu^-_{m^-(\mu)}, -)
\label{mu_sequence}
\ee
\item
\be
 \epsilon_j(\mu)= + \ \text{ if } \alpha_j \in S, \ \text{ and } \epsilon_{r+j}(\mu)= - \ \text{ if } \beta_j +1\in S, \quad j=1, \dots, r.
 \label{canon_epsilon_constraint}
 \ee
\end{enumerate}
\end{lemma}

\begin{definition}
 The {\em sign} of the polarization $(\mu^+, \mu^-)$, denoted $\sgn(\mu)$, is defined 
as the sign of the permutation  that takes the sequence (\ref{alpha_beta_sequence})  into the sequence
(\ref{mu_sequence}).
\end{definition}
Denote by
\be
\pi(\mu^\pm) :=\#(\alpha \cap \mu^{\pm})
\label{pi_pm_mu_def}
\ee
the cardinality of the intersection of $\alpha$ with $\mu^{\pm}$. It follows that
\be
\pi(\mu^+) + \pi(\mu^-) = r+s. 
\label{pi_pm_sum}
\ee
 We then have
 \begin{lemma}
\be
|\lambda \rangle =  {(-1)^{\tfrac{1}{2}r(r+1) + s}\over 2^{r-s}} \sum_{\mu\in \PP(\alpha, \beta)}
\sgn(\mu)(-1)^{\pi(\mu^-)} i^{ m(\mu^-)}\prod_{j=1}^{m(\mu^+)} \phi^+_{\mu^+_j} \prod_{k=1}^{m(\mu^-)} \phi^-_{\mu^-_k}  |0 \rangle.
\label{polariz_sum}
\ee
\end{lemma}
\begin{proof}
 In  eq.~(\ref{lambda_frob_vac}), reorder the product over the factors $\psi_{\alpha_j} \psi^\dag_{-\beta_j-1}$,   so 
 the  $\psi_{\alpha_j}$ terms precede the $\psi^\dag_{-\beta_j-1}$ ones, giving an overall sign 
 factor $(-1)^{\tfrac{1}{2}r(r-1)}$.  Then substitute 
 \be
\psi_{\alpha_j} = \tfrac{1} {\sqrt{2}}( \phi^+_{\alpha_j} - i \phi^-_{\alpha_j}), \quad
 \psi^\dag_{-\beta_j-1}= \tfrac{(-1)^j}{ \sqrt{2}} (\phi^+_{\beta_j+1} + i \phi^-_{\beta_j+1}), \quad j\in \Zb,
\ee
which follows from (\ref{neutral_charged}),
for all the factors, and expand  the product as a sum over monomial terms of the form
\be
 \prod_{j=1}^{m(\mu^+)} \phi^+_{\mu^+_j} \prod_{k=1}^{m(\mu^-)} \phi^-_{\mu^-_k}  |0 \rangle.
 \ee
 Taking into account the sign factor $\sgn(\mu)$ corresponding to the order of the neutral fermion
factors, as well as the powers of $-1$ and $i$, and noting that there are $2^s$  resulting identical terms, 
then gives (\ref{polariz_sum}).
\end{proof}

\begin{definition}{\bf Supplemented partitions.}
If $\nu$ is a strict partition of cardinality $m(\nu)$ (with $0$ allowed as a part), 
we define the associated {\em supplemented partition} $\hat{\nu}$ to be
\be
\hat{\nu} := \begin{cases}  \nu \ \text{ if } m(\nu) \ \text{ is even}, \cr
                  (\nu,0)  \  \text{ if } m(\nu) \  \text{ is odd}.
                  \end{cases}
 \label{hat_nu}
\ee
We denote by $m(\hat{\nu})$  the cardinality of $\hat{\nu}$.
\end{definition}

\subsection{Bilinear expansion theorem}
\label{KP_BKP_bilinear}

Combining the above, we arrive at our main  result. 
\begin{theorem}
\label{KP_BKP_lattice_bilinear}
Restricting the group element $\hat{g}$ to $\hat{g}(h)=\hat{h}^+ \hat{h}^-$ and the $\tb$ values to $\tb'$, 
the lattice of KP $\tau$-functions $\pi_{(\alpha|\beta)}(g(h))({\bf t}')$ may be expressed as sums over products
of pairs of BKP $\tau$-functions $\kappa_{\hat{\mu}^+} (h)({\bf t}_B)$ and $\kappa_{\hat{\mu}^-} (h)({\bf t}_B)$ as follows
\be
\pi_{(\alpha|\beta)}(g(h))({\bf t}')= {(-1)^{\tfrac{1}{2}r(r+1) + s}\over 2^{r-s}}\sum_{\mu\in \PP(\alpha, \beta)}
\sgn(\mu)(-1)^{\pi(\mu^-) +\tfrac{1}{2}m(\hat{\mu}^-)}\kappa_{\hat{\mu}^+} (h)({\bf t}_B)\kappa_{\hat{\mu}^-} (h)({\bf t}_B).
\label{bilinear_exp}
\ee
\end{theorem}
\begin{proof}
Substituting (\ref{polariz_sum}) into (\ref{pi_alpha_beta}), using (\ref{kappa_alpha}) and (\ref{alpha_pm}) and
applying the factorizations (\ref{g_h_def}), (\ref{gamma_factoriz}) and Lemma \ref{factorization_lemma}
 gives the bilinear expansion \ref{bilinear_exp}.
\end{proof}

\section{Examples}
\label{special_cases}
The first two examples, \ref{det_pfaff_bilin_sum} and \ref{schur_Q_schur_bilin_sum},  show how the results 
of \cite{BHH} and \cite{HO} are recovered by choosing special values for the parameters in Theorem \ref{KP_BKP_lattice_bilinear} . 
Examples (\ref{symm_polynom_bilin_sum})  and (\ref{r_p_polynom_bilin_sum}) give families of lattices of  
 polynomial  $\tau$-functions of KP and BKP type that generalize the 
Schur functions and Schur $Q$-functions \cite{BL, HL, KvdL1, KvdL2, KvdLRoz, Iv}. 
These include, as special cases: {\em factorial}, {\em shifted}  and {\em interpolation Schur functions}, \cite{BL, OkOl, Ol} 
{\em interpolation} Schur $Q$-functions \cite{Iv}, and multivariate Laguerre polynomials \cite{La, Ol}.


\begin{example}[\bf Determinants of skew matrices as bilinear sums over Pfaffians]
\label{det_pfaff_bilin_sum}

The bilinear identity in \cite{BHH} follows as a particular case of Theorem \ref{KP_BKP_lattice_bilinear},
 Setting \hbox{${\bf t} = {\bf 0} = (0, 0, \dots)$} in (\ref{bilinear_exp}),
  the $\tau$-function $\pi_{\alpha |\beta)}({\bf 0})(g)$ becomes the  $(\alpha|\beta)$ Pl\"ucker coordinate
of the element $g^T(\HH)$ of the Grassmannian $\Gr_{\HH}(\HH+\HH^*)$ of subspaces of $\HH+\HH^*$ commensurate
with the subspace $\HH$ which is mapped, under the Pl\"ucker map, to the vacuum $|0\rangle$.

In the big cell of the Grassmannian, which consists of those elements that are the graph of a map $w: \HH\ra \HH^*$,
it is the determinant of minors of the  affine coordinate matrix $A$ which, for $g$ of the form
(\ref{g_h_def}) is skew symmetric. And $\kappa_{\hat{\mu}^\pm} (h)({\bf 0}_B)$ are the {\em Cartan coordinates} \cite{BHH} of
the element $h^T(\HH)$ of the maximal  isotropic Grassmannian $\Gr^0_\HH(\HH+\HH^*, Q)$, which is the
Pfaffian of the principal minors of the affine coordinate matrix $A$ corresponding to the strict partitions $\mu^\pm$.
This reproduces the expansions of the determinants of minors of a skew matrix as sums over products of Pfaffians of their
principal minors given by Corollary 6.4 of \cite{BHH}. 
\be
 \det(A_{(I | J)}) = {(-1)^{\frac{1}{2}r(r+1)+s}\over 2^{r-s} } {\hskip -10 pt}\sum_{(K,L) \in \PP(I, \beta)}
   {\hskip -10 pt} (-1)^{\pi +l/2} \sgn(K,L) \Pf(A_{(K | K)})\Pf( A_{(L | L)}), 
 \label{single_sum_pfaff_CB_ident}
 \ee
 where $A$ is a skew symmetric $N \times N$ matrix, $I$ and $J$ are the rows and columns
 \be
 I=(I_1, \dots, I_r) =I(\alpha) , \quad J=(J_1, \dots, I_r)=I(\beta)
 \label{I_J_def}
 \ee
of the $r \times r$ minor $A_{(I | J)}$, 
\be
K=(K_1, \dots, K_k)= I(\mu^+), \quad L=(L_1, \dots, L_l)= I(\mu^-)
\label{K_L_def}
\ee
are the rows and columns of the pair of principal minors $A_{(K| K)}$ and $A_{(L | L)}$, and
\be
\pi:=\#(I\cap L),\quad k=m(\mu^+)=\#(K), \quad  l=m(\mu^-)=\#(L), 
\label{k_l_def_sum}
\ee
The ordered subsets $I,J, K, L \ss (1, \dots N)$ are polarizations related by
\be
I\cap J = K\cap L:=S, \quad I\cup J = K\cup L.
\label{I_cap_cupJ}
\ee
$s$  is the cardinality of $S$, and both $k$ and $l$ must have even parity
 in order that the Pfaffians in the sum not vanish.

\end{example}


 \begin{example}[\bf Schur functions as bilinear sums over Schur $Q$-functions]
\label{schur_Q_schur_bilin_sum}

If we set $h$ equal to the identity element $\Ib\in \SO(\HH+\HH^*, Q)$, the  $\tau$-function $\pi_{\alpha|\beta)}({\bf t}')(\Ib)$
becomes the Schur function  \cite{DJKM1, DJKM2, JM, HB}
\be
s_{(\alpha|\beta)}({\bf t}) =  \pi_{\alpha|\beta)}(\Ib)({\bf t}) =\langle 0 | \hat{\gamma}_+({\bf t})| \lambda \rangle, 
\label{schur_fermionic}
\ee
which, by the Giambelli identity (which is a particular case of (\ref{pi_alpha_beta_det})), is the determinant 
of the matrix formed from Schur functions corresponding to hook partitions, whose restriction 
to ${\bf t}'$ equals the skew matrix formed from the Schur $Q$-functions $Q_{jk}({\bf t}_B)$ with  Frobenius rank $2$.
 Then $\kappa_{\hat{\mu}^\pm} (\Ib_{\HH^\pm_0})({\bf t}_B))$ become the  (scaled) Schur $Q$-functions 
 \bea
\QQ_{\hat{\mu}^\pm}({\bf t}_B)&\&= \kappa_{\hat{\mu}^\pm} (\Ib_{\HH^\pm_0})({\bf t}_B))
= \langle 0 |  \hat{\gamma}^{B\pm} ({\bf t}_B) | \hat{\mu}^\pm) \cr
&\&= 2^{-\frac{1}{2}m(\hat{\mu}^\pm)} \tilde{Q}_{\hat{\mu}^\pm}(\tfrac{1}{2} \tb_B),
 \eea
 where, restricting the independent variables to $\tb_B = [\xb]_B$ as in (\ref{t_norm_power_sum}), we recover the usual Q Schur polynomials
 \be
 Q(\xb) = \tilde{Q}(\tb_B),
 \ee
 which, by (\ref{kappa_alpha_pfaff}), are the Pfaffians of the principal minors of the same skew symmetric matrix 
 $\QQ_{\mu^\pm_j, \mu^\pm_k}({\bf t}_B)$. 
 
 Eq.~(\ref{bilinear_exp}) therefore reduces to the  result obtained in \cite{HO}, 
 expressing Schur functions as sums over products of Schur $Q$-functions,
 \be
s_{(\alpha|\beta)}(\tb')  = {(-1)^{\tfrac{1}{2}r(r+1)+s} \over 2^{2r-s}} \sum_{\mu\in \PP(\alpha, \beta)} \sgn(\mu)
 (-1)^{\pi(\mu^-)+\tfrac{1}{2}m(\hat{\mu}^-)}
 \tilde{Q}_{\hat{\mu}^+}(\tfrac{1}{2}\tb_B) \tilde{Q}_{\hat{\mu}^-}(\tfrac{1}{2}\tb_B).
   \label{s_lambda_Q_bilinear_bis}
 \ee
 \end{example}
 
 
 \begin{example}{\bf (Lattices of polynomial  KP vs. BKP $\tau$-functions)}.
\label{symm_polynom_bilin_sum}
A broad family of KP and BKP $\tau$-function lattices consisting of
more general (inhomogeneous) polynomials in the flow variables
is obtained if special restrictions are placed on the matrices $\tilde{A}$  relating them to those
appearing in (\ref{hat_g_exp_A}) and (\ref{hat_h_alpha_pm_exp_A}). 
Such  polynomial KP $\tau$-functions, their BKP analogs  and their fermionic representations were studied 
in \cite{Iv, HL, KvdL2, KvdLRoz, Roz}. They may be viewed as generalized Schur functions and Schur $Q$-functions,
which are expressible as finite linear combinations of ordinary Schur and Schur $Q$-functions.
We introduce the following fermionic operators
\be
\psi_i(\tilde{A})=\hat{g}(\tilde{A})\psi_i \left( \hat{g}(\tilde{A)}\right)^{-1},\quad \psi_i^*(\tilde{A})=\hat{g}(\tilde{A})\psi_i^\dag \left( \hat{g}(\tilde{A})\right)^{-1}
\label{psi_A}
\ee
where $\hat{g}(\tilde{A})$ is defined in (\ref{hat_g_exp_A}), and
\be
\phi^\pm_i(A)=\hat{h}^\pm(A)\phi_i^\pm (\hat{h}^\pm(A))^{-1},
\label{phi_pm_A}
\ee
where  $\hat{h}^\pm(A)$ is defined in (\ref{hat_h_alpha_pm_exp_A}).
\begin{lemma}
If the infinite skew matrix  $A$ is chosen to satisfy the antidiagonal triangular condition:
\be
A_{j, -k} = 0 \quad \text{if } j>k
\ee
and hence $\tilde{A}$ is has the upper  triangular form
\be
\tilde{A}_{jk} = 
\begin{cases} (-1)^k A_{j, -k} \ \text{ if } \ j\le k \cr
0  \ \text{ if } \ j > k, 
\end{cases}
\ee
then
\be
\hat{g}(\tilde{A}) = \hat{h}^+(A) \hat{h}^-(A) 
\ee
and
\be
\psi_j(\tilde{A})=\tfrac{1}{\sqrt{2}}(\phi_j^+(A)- i \phi_j^-(A)),\quad \psi_{-j}^*(\tilde{A})=\tfrac{(-1)^j}{\sqrt{2}}(\phi_j^+(A)+ i \phi_j^-(A))
\ee
\end{lemma}
\begin{proof}
This follows from the definitions (\ref{hat_g_exp_A}) and \ref{hat_h_alpha_pm_exp_A}) of   $\hat{g}(\tilde{A})$ and $\hat{h}^\pm(A)$,
the relation (\ref{neutral_charged}) between charged and neutral fermionic operators, the skew symmetry of the matrix $A$,
which leads to the cancellation of the sum over mixed terms $\phi^+_j\phi^-_k$ in the exponent,
and the mutual anticommutatitvity  of the two types of neutral fermi operators $\{\phi^+_j\}_{j\in \Zb}$ and $\{\phi^-_j\}_{j\in \Zb}$
\end{proof}

To define a lattice of polynomial KP $\tau$-functions that generalize the Schur polynomials we note that,
because  $\tilde{A}$ is upper triangular, $\hat{g}(h(A))$  stabilizes the vacuum $|0\rangle$
\be
 \hat{g}(\tilde{A})|0\rangle =|0\rangle.
 \label{g_tilde_A_vac}
 \ee
It follows  from (\ref{pi_alpha_beta_det}) that
\be
\pi_{(\alpha|\beta)}(g(h(A)))({\bf t})  =
\det \left( \langle 0|\hat{\gamma}_+(\tb)\hat{g}(\tilde{A}) \psi_{\alpha_i}\psi^\dag_{-\beta_j-1}|0\rangle \right)_{i,j=1,\dots,r}
= \langle 0|\hat{\gamma}_+(\tb)\hat{g}(\tilde{A})|\lambda\rangle.
\label{pi_alpha_beta_def}
\ee
Since both $\hat{\gamma}_+(\tb)$ and $\hat{g}(\tilde{A})$ stabilize the vacuum, we may define
\bea
s_\lambda({\tb}|\tilde{A})&\& := \pi_{(\alpha|\beta)}(g(h(A)))({\bf t}) \cr
&\&= (-1)^{\sum_{j=1}^r\beta_j}(-1)^{\tfrac{1}{2}r(r-1)} \langle 0|\hat{\gamma}_+(\tb)\psi_{\alpha_1}(\tilde{A})\cdots \psi_{\alpha_r}(\tilde{A}) \psi^*_{-\beta_1-1}(\tilde{A})\cdots \psi^*_{-\beta_r-1}(\tilde{A})|0\rangle
\cr
&\&= (-1)^{\sum_{j=1}^r\beta_j}\det\left( \langle 0|\hat{\gamma}_+(\tb)\psi_{\alpha_i}(\tilde{A}) \psi^*_{-\beta_j-1}(\tilde{A})|0\rangle \right)_{i,j=1,\dots, r}, 
\label{gen_Schur_def}
\eea
and these can be interpreted as generalized Schur functions (cf.~eqs.~(\ref{lambda_frob_vac}), \ref{schur_fermionic})).

Note that,  as for ordinary Schur functions \cite{Mac}, we have, for any ${\bf p}=(p_1, p_2, \dots )$, 
 \be
 s_{\lambda}({\tb +{\bf p}}|\tilde{A})=\sum_{\rho\subseteq \lambda} 
  s_{\lambda/\rho}({\bf p}|\tilde{A}) s_{\rho}(\tb), 
  \label{s_lambda_A_p_shifted}
 \ee
 where the generalized skew Schur function is defined fermionically as
\be
s_{\lambda/\rho}({\tb}|\tilde{A}):= (-1)^{\sum_{j=1}^r\beta_j}(-1)^{\tfrac{1}{2}r(r-1)}
\langle \rho|\hat{\gamma}_+(\tb)\psi_{\alpha_1}(\tilde{A})\cdots 
\psi_{\alpha_r}(\tilde{A})\psi^*_{-\beta_1-1}(\tilde{A}) \cdots \psi^*_{-\beta_1-1}(\tilde{A})|0\rangle.
\label{general_schur_translate}
\ee
This follows from factorizing
 \be
 \hat{\gamma}_+({\bf t} + {\bf p}) =  \hat{\gamma}_+ ({\bf t})  \hat{\gamma}_+({\bf p} )
 \ee
 in (\ref{pi_alpha_beta_def}) and inserting the projection operator 
 \be
 \sum_{\rho} | \rho;0 \rangle \langle \rho;0 |
 \ee
onto the  $n=0$ fermionic charge sector between the factors.
It follows from (\ref{general_schur_translate}) that   
$s_{\lambda}({\tb}|\tilde{A})$ is, indeed, a polynomial in the variables ${\bf t}=(t_1, t_2, \dots)$
since setting ${\bf p} ={\bf 0}$ in (\ref{s_lambda_A_p_shifted})  gives
\be
 s_{\lambda}({\tb }|\tilde{A})=\sum_{\rho\subseteq \lambda}  s_{\lambda/\rho}({\bf 0}|\tilde{A}) s_{\rho}(\tb).
 \label{s_lambda_A_expan_s_rho}
 \ee

Next,  we define the (scaled) generalized Schur $Q$-functions to be
\bea
\QQ_{\hat{\mu}^\pm}({\tb_B}|A)&\&:= \kappa_{\alpha} (h(A))({\bf t}_B) =
\langle 0|\gamma^{B\pm}(\tb_B)\phi^\pm_{\hat{\mu}^\pm_1}(A)\cdots \phi^\pm_{\hat{\mu}^\pm_{m(\hat{\mu}^\pm)}}(A) |0\rangle
 \cr
&\&=\Pf \left( \langle 0|\gamma^{B^\pm}(\tb_B)\phi^\pm_{\hat{\mu}^\pm_i}(A^\pm)\phi^\pm_{\hat{\mu}^\pm_j}(A)|0\rangle \right)_{i,j=1,\dots, m(\hat{\mu}^\pm)} \cr
&\& =: 2^{-\tfrac{1}{2} m(\hat{\mu}^\pm)} \tilde{Q}_{\hat{\mu}^\pm}(\tfrac{1}{2}{\tb_B}|A)
\label{gen_Q_def}
\eea
(Recall that $\hat{\mu}^\pm$ denotes the {\em  supplemented partition}, 
of even cardinality, which is obtained from  $\mu^\pm$ by adding a $0$ part if it has odd cardinality.)
These too are polynomials in the BKP flow variables $\tb_B$ since, like (\ref{s_lambda_A_expan_s_rho}), we 
may express $\QQ_{\hat{\mu}^\pm}({\tb_B}|A)$ as a finite linear combination of the Schur $Q$-functions
$\QQ_{\hat{\nu}}({\tb_B})$.
\be
\QQ_{\hat{\mu}^\pm}(\tb_B | A)({\bf p}_B)) =
\sum_{\nu \subseteq \hat{\mu}^\pm}\, 
 \,\QQ_{\hat{\mu}^\pm/\nu}({\bf 0}|A)\,\QQ_{\nu}({\tb_B}),
\ee

 Restricting to this particular case of Theorem \ref{KP_BKP_bilinear} gives:
\begin{corollary} 
\label{bilinear_general_schur}
The generalized Schur function $s_{(\alpha|\beta)}$  (\ref{gen_Schur_def}), evaluated  at $\tb=\tb'$, 
may be  expressed as the following sum  over products $ \tilde{Q}_{\hat{\mu}^+} \tilde{Q}_{\hat{\mu}^-}$ of pairs of generalized 
$Q$ Schur functions (\ref{gen_Q_def}):
 \be
s_{(\alpha|\beta)}(\tb'|A)  = {(-1)^{\tfrac{1}{2}r(r+1)+s} \over 2^{2r-s}} \sum_{\mu\in \PP(\alpha, \beta)} \sgn(\mu)
 (-1)^{\pi(\mu^-)+\tfrac{1}{2}m(\hat{\mu}^-)}
 \tilde{Q}_{\hat{\mu}^+}(\tfrac{1}{2}\tb_B|A) \tilde{Q}_{\hat{\mu}^-}(\tfrac{1}{2}\tb_B|A).
   \label{s_alpha_beta_Q_bilinear}
 \ee
 \end{corollary}
 \begin{remark}
To interpret the results in terms of symmetric polynomials in a set of auxiliary variables,
 and consistently be able to choose the specialization $\tb = \tb'$ in (\ref{s_alpha_beta_Q_bilinear}),
we must reinterpret them as supersymmetric power sums,  as in (\ref{t_prime_y_sum}).
 \end{remark}
 
A particularly simple case  consists of the double $D(\alpha)$ of a strict partition $\alpha$, 
for which the Frobenius indices $(\alpha |\beta)$ are related by $\alpha= I(\beta)$.
Since the only polarization in this case is itself, the sum consists of just one term, and we obtain
(as for ordinary Schur functions)
 \be
 s_{D(\alpha)}(\tb'|A)= 2^{-r} \left( \tilde{Q}_{{\hat\alpha}}(\tfrac{1}{2}\tb_B|A^\pm) \right)^2.
 \ee
\end{example}

 
 \begin{example}{\bf (Lattices of polynomial $\tau$-functions parametrized by ${\bf r}$ and ${\bf p}$)}.
\label{r_p_polynom_bilin_sum}
An interesting  subclass of generalized Schur functions and Schur $Q$-functions
is obtained if we choose the skew matrix $A$ to be parametrized in terms    
two sets of  infinite parameters ${\bf r}:=\{r_j\}_{j \in \Zb}$ and  ${\bf p}=(p_1,p_2,p_3,\dots)$,
where we assume the relations
\be
r_j= r_{1-j}, \quad j\in \Nb.
\ee

Let 
\bea
\hat{A}^\rb_j &\&:=\sum_{k\in\Zb} \psi_k \psi^\dag_{k+j} r_{k+1}\cdots r_{k+j}, \quad j=1,2,3,\dots
\label{hat_A_r_j} \\
\hat{A}^{\rb\pm}_j&\&:=\sum_{k\in\Zb}(-1)^{k+1} \phi^\pm_k \phi^\pm_{-k-j} r_{k+1}\cdots r_{k+j},\quad j=1,3,5,\dots
\label{hat_A_j^pm} \\
\hat{A}^\rb(\pb) &\&:=\sum_{j>0}p_j A^\rb_j =: \sum_{j, k  \in \Zb} \tilde{A}^\rb(\pb)_{jk} \no{\psi_j \psi^\dag_k}, \\
 \hat{A}^{\rb \pm}(\pb_B)&\&:=\sum_{j>0,\,{\rm odd}}p_j\hat{A}^{\rb \pm}_j =: \tfrac{1}{2}\sum_{j,k \in \Zb}  A^{\rb}(\pb_B)_{jk} \no{\phi^\pm_j \phi^\pm_{k}} 
\label{hat_A_r_p}
\eea
where
\bea
 \tilde{A}^\rb(\pb)_{jk}&\&= p_{k-j} r_{j+1} \cdots r_k, \quad k>j 
 \label{A_r_p_jk},
   \\
 &\&\cr
A^{\rb}(\pb_B)_{jk} &\&=\begin{cases}  (-1)^k \, p_{-j-k}\, r_{j+1} \cdots r_{-k} , \quad j+k < 0 
  \text{ and odd}\cr
 \qquad  0 \quad \text{ if } j+k \ \text{ is even or } j+k \ge 0.
  \end{cases}
  \cr
  &\&
 \label{A_r_p_B_jk}
\eea
and define the group elements
\bea
\hat{g}(\tilde{A}^\rb(\pb)) := e^{\hat{A}^\rb({\bf p}) }, \quad \hat{h}^\pm(A^{\rb}(\pb_B)):= e^{\hat{A}^{\rb\pm }(\pb_B) }.
\label{g_Ar_h_Ar_def}
\eea
 We then have:
 \begin{proposition}
\be
s_\lambda(\tb | \tilde{A}^\rb(\pb))  = \pi_{(\alpha|\beta)}(g(\tilde{A}^\rb(\pb)))({\bf t}) =
\,\sum_{\rho \subseteq \lambda}
\,r_{\lambda/\rho}\,s_{\lambda/\rho}({\bf p})\,s_{\rho}({\bf t}),
\label{s_lambda_r_p_expansion}
\ee
where $s_{\lambda/\rho}({\bf p})$ is the skew Schur function corresponding to the skew
partition $\lambda/\rho$, and
\be
 r_{\lambda/\rho} :=\prod_{i=1}^{\ell(\lambda)} \prod_{j=\rho_i+1}^{\lambda_i}r_{j-i} ={ r_\lambda \over r_\rho}
 \label{skew_content_prof_r}
\ee 
is the content product over nodes of the skew  Young diagram of $\lambda/\rho$, and
\be
\QQ_{\hat{\mu}^\pm}(\tb_B | A^{\rb}({\bf p}_B)) = \kappa_{\hat{\mu}^\pm}(g(A^{\rb} (\pb_B)))({\bf t}_B) =
\sum_{\nu \subseteq \hat{\mu}^\pm}\, 
r^B_{\mu^\pm /\nu} \,\QQ_{\hat{\mu}^\pm/\nu}({\bf p}_B)\,\QQ_{\nu}({\tb_B}),
\label{Q_mu_r_p_expansion}
\ee
 where 
\be
\QQ_{\hat{\mu}^\pm/\nu}({\bf p}_B) = (\nu^\pm|\gamma^{B\pm}(\pb_B) | \hat{\mu}^\pm)
\ee
  is the (scaled) skew Schur Q-function 
 corresponding to the skew partition $\mu^\pm/\nu$ 
  and
\be
 r^B_{\mu^\pm/\nu}   :=\prod_{i=1}^{m(\mu^\pm)} \prod_{j=\nu_i+1}^{\mu^\pm_i} r_j ={ r^B_{\mu^\pm} \over r^B_\nu}
 \label{skew_strict_content_prod_r}
\ee 
is the content product over nodes of the {\em shifted} skew Young diagram  \cite{Iv} of $\mu^\pm/\nu$.

 \end{proposition}
\begin{proof}This follows along the lines of  \cite{Or1, OS}.
Define the operators
\be
\hat{C}^\rb:=\exp \left( \sum_{j< 0} T^\rb_{j}\psi^\dag_j\psi_j -\sum_{i\ge 0} T^\rb _j\psi_j\psi^\dag_j \right) 
\label{C_r}
\ee
where
\be
r_j=e^{T^\rb_{j-1}-T^\rb_{j}}, \quad T^\rb_0:=0.
\label{r_T}
\ee
Imposing the restriction
\be
T^\rb_j=-T^\rb_{-j},
\label{B-reduction}
\ee
 implies
\be
r_j=r_{1-j}.
\label{B-reduction-r}
\ee
We also define 
\be
\hat{C}^{\rb\pm} := \exp \sum_{j\ge 0} (-1)^{j+1}T^\rb_j\phi^\pm_j\phi^\pm_{-j} .
\label{C_r_pm}
\ee
It then follows from  (\ref{neutral_charged}) that
\be
\label{C_C_pm}
\hat{C}^\rb=\hat{C}^{\rb+}\hat{C}^{\rb-}=\hat{C}^{\rb-}\hat{C}^{\rb+}.
\ee

From (\ref{psi_i_anticomm}) and (\ref{neutral_ccr}) we get
\be
\hat{C}^\rb\psi_j (\hat{C}^\rb)^{-1}= e^{-T^\rb_j}\psi_j,\quad \hat{C}^\rb \psi^\dag_j (\hat{C}^\rb)^{-1}= e^{T^\rb_j}\psi^\dag_j, \quad 
\hat{C}^{\rb\pm} \phi^\pm_j (\hat{C}^{\rb\pm})^{-1}= e^{-T^\rb_j}\phi^\pm_j
\ee
It follows that the operators defined in (\ref{hat_A_r_j}, (\ref{hat_A_j^pm}), (\ref{hat_A_r_p}) are given by
\be
\hat{A}^\rb_j: = \hat{C}^\rb J_j(\hat{C}^\rb)^{-1},\quad j=1,2,\dots
\ee
and
 \be
\hat{A}^{\rb\pm}_j :=\hat{C}^{\rb\pm} J_j^{B\pm} (\hat{C}^{\rb\pm})^{-1},\quad j=1, 3, 5,\dots
\ee
For odd $j$  we have
\be
\hat{A}^\rb_j  =  \hat{C}^\rb(J_j^{B+}+J^{B-}_j) (\hat{C}^\rb)^{-1}=\hat{A}_j^{\rb+} + \hat{A}_j^{\rb-},\quad i=1,3,5,\dots
\ee

Defining
\be
\hat{\gamma}^{\rb}_+({\bf p}):=(\hat{C}^\rb)^{-1}\hat{\gamma}_+({\bf p})\hat{C}^\rb,\quad 
\hat{\gamma}^{\rb,B\pm}({\bf p}_B):= (\hat{C}^{\rb\pm})^{-1}\hat{\gamma}^{B\pm}({\bf p}_B) \hat{C}^{\rb \pm}
\ee
it follows from  (\ref{C_C_pm}) and (\ref{gamma_factoriz}) that
\be
\label{gamma_r_factoriz}
\hat{\gamma}_+^\rb({\bf p}')=\hat{\gamma}^{\rb ,B+}({\bf p}_B)\gamma^{ \rb, B-}({\bf p}_B).
\ee
The equalities
\be
s_\lambda({\bf t}|A^\rb({\bf p}))=\langle 0|\hat{\gamma}_+(\tb)\hat{\gamma}^\rb_+({\bf p})|\lambda\rangle
\label{gen-S=VEV}
\ee
and
\be
\QQ_{\mu^\pm}({\bf t}_B|A^{\rb}({\bf p}_B))=\langle 0|\hat{\gamma}^\pm(\tb_B)\hat{\gamma}^{\rb, \Bb\pm}({\bf p}_B)|\mu^\pm\rangle.
\label{Q_r_p_fermi}
\ee
follow from  definitions (\ref{gen_Schur_def}) and (\ref{gen_Q_def}), when
the group elements  $\hat{g}(h(A))$ and $\hat{h}^\pm(A)$  are chosen as in (\ref{g_Ar_h_Ar_def}) and $\psi_i(\tilde{A})$
and $\phi^\pm_j(A)$ computed by inserting these in (\ref{psi_A}) and (\ref{phi_pm_A}).

Inserting the projection operator 
\be
\Pi_0=\sum_{\rho\in\Pa} |\rho; 0\rangle \langle\rho; 0|
\ee
onto the subspace $\FF_0\ss\FF$  between $\hat{\gamma}_+({\bf t})$ 
and $\hat{\gamma}^\rb_+({\bf p})$  in (\ref{gen-S=VEV}), and using 
\be
s_\lambda(\tb)=\langle 0|\hat{\gamma}_+(\tb)|\lambda\rangle
\ee
and the slightly more general relation
\be
r_{\lambda/\rho}s_{\lambda/\rho}({\bf p})= \langle \rho|\hat{\gamma}_+^{\rb}({\bf p}) |\lambda \rangle 
\ee
gives (\ref{s_lambda_r_p_expansion}). (See formula (3.5.7) in  \cite{Or2}) 

Similarly, to prove (\ref{Q_mu_r_p_expansion}),  insert the projection operators 
\be
\Pi^\pm = \sum_{\nu^{\pm}}  |\nu^{\pm} ) ( \nu^{\pm} |   
\ee  
onto the subspaces $\tilde{\FF}^\pm_\phi \ss \FF$  between
$\hat{\gamma}^{B\pm}(\tb_B)$ and $\gamma^{\rb, B\pm}(\pb_B)$  in (\ref{Q_r_p_fermi}) 
and use
\be
\QQ_{\mu^{\pm}}({\bf t}_B)=
\langle 0| \hat{\gamma}^{B\pm}({\bf t}_B)|\mu^{\pm} )
\label{Q-You}
\ee
and
\be
\label{skewQ-fermionic}
r^B_{\mu^\pm /\nu}\QQ_{\mu^\pm/\nu}({\bf p}_B)=(\nu^\pm| \hat{\gamma}^{\rb, B\pm}({\bf p}_B)|\mu^\pm).
\ee
\end{proof}
It then follows  from (\ref{gamma_r_factoriz}) that  the polynomials 
$s_\lambda({\bf t}'|A^\rb({\bf p}'))$ and $ \tilde{Q}_{\hat{\mu}^\pm}(\frac{1}{2}{\bf t}_B|A^{\rb \pm}({\bf p}_B))$ are related
by (\ref{s_alpha_beta_Q_bilinear}), with $\tilde{A}= \tilde{A}^\rb(\pb')$ and $A= A^\rb(\pb_B)$.
\begin{remark}
With the special choice 
\bea
\pb' &\& =\tb_0:=(1,0,0,\dots)
\label{p_0_def}
\\
t_j &\&={1\over j}\sum_{k=1}^N x_k^j
\label{t-x}
\cr
r_j&\&=r_j(z,z'):=-(z+j)(z'+j),
\label{r_j_laguerre}
\eea
eq.~(\ref{s_lambda_r_p_expansion}) gives the multivariate Laguerre polynomial  \cite{Ol, La}. 
To see this, substitute (\ref{p_0_def}),  (\ref{r_j_laguerre}) into eq.~(\ref{s_lambda_r_p_expansion}) to obtain 
\be
s_\lambda(\tb|A^{\rb(z,z')}(\tb_0))= 
\sum_{\rho\subseteq\lambda} (-1)^{|\lambda|-|\rho|}(z)_{\lambda/\rho}(z')_{\lambda/\rho}
s_{\lambda/\rho}(\tb_0)s_{\rho}(\tb),
\label{multivar_laguerre_exp}
\ee
where
\be
(z)_{\lambda/\rho}:=\prod_{(i,j)\in\lambda/\rho }(z+j-i)
\ee
is the content product of $\{z+j\}_{j\in \Zb}$ over the skew Young diagram  $\lambda/\rho$, and 
\be
s_{\lambda/\rho}(\tb_0)= \det\left({1\over (\lambda_i - i -\rho_j +j)!}\right)_{1\le i, j \le \ell(\lambda)}.
\ee
Eq.~(\ref{multivar_laguerre_exp}) coincides with formula (4.3) of \cite{Ol} for multivariate Laguerre polynomials.
where the coefficients were also written in terms of integer evaluations of shifted Schur functions \cite{OkOl}.
\end{remark}

Note that  an arbitrary choice of parameters  $(z,z')$  in  formula (\ref{multivar_laguerre_exp}) does not satisfy
 the relation (\ref{g_h_def}), which is needed in order to apply  Corollary \ref{bilinear_general_schur}. 
For this, the reduction condition (\ref{B-reduction-r}) must be satisfied,
which requires that we impose the constraint \hbox{$z'=-1-z$}, giving
\be
r_j=(z+j)(z+1-j) = r_{1-j}.
\ee
Substituting this into (\ref{skew_strict_content_prod_r}), we get eq.~(\ref{Q_mu_r_p_expansion}), with
\be
r^B_{\mu^\pm/\nu}(z) :=\prod_{i=1}^{m(\nu)}\prod_{j=\nu_i+1}^{\mu_i}(z+j)(1+z-j)
\ee
and 
\be
\tilde{Q}_{\mu/\nu}(\tfrac 12\tb_0)
= \Pf \left(\frac{\mu_i-\mu_j-\nu_i+\nu_j}{(\mu_i+\mu_j-\nu_i-\nu_j)(\mu_i-\nu_i)!(\mu_j-\nu_j)!}\right)_{1\le i,j \le m(\mu)},
\ee
with the understanding that $\nu_i=0$ if $i > m(\nu)$. 
 \end{example}
 
\medskip
\noindent 
\small{{\em Acknowledgements.} The work of J.~H. was supported by the Natural Sciences and Engineering Research Council of Canada (NSERC); that of A.~Yu.~O.  by RFBR grant 18-01-00273a and state assignment  N 0149-2019-0002. \hfill

\begin{center}
\smallskip
This work is dedicated to {\bf Grigori Olshanski} on the occasion of his 70th birthday. \hfill
\end{center}



\end{document}